\documentclass[authoryear,final,3p,times,english]{elsarticle}
\usepackage{lineno,hyperref}
\usepackage{amsmath}
\usepackage{amsthm}
\usepackage{amsfonts}
\usepackage{natbib}
\usepackage[bf]{caption}
\usepackage{caption}
\usepackage[T1]{fontenc}
\usepackage[utf8]{inputenc}
\usepackage{babel}
\usepackage[inline]{enumitem}
\usepackage{a4,graphicx}
\usepackage{subfig}	
\usepackage{float}
\usepackage{array}
\usepackage{setspace}
\usepackage{geometry} 
\usepackage{array}
\newtheorem{proposition}{Proposition }[section]
\newtheorem{remark}{Remark }[section]
\usepackage{varwidth}
\usepackage{mathrsfs}

\usepackage{epstopdf}
\usepackage[toc,page]{appendix}

\theoremstyle{definition}

\usepackage{algorithm}
\usepackage{multicol}
\usepackage{arydshln}
\usepackage{caption}
\usepackage{algpseudocode}
\hypersetup{colorlinks=true,linkcolor=green,filecolor=blue,urlcolor=blue}
\makeatletter
\def\ps@pprintTitle{%
	\let\@oddhead\@empty
	\let\@evenhead\@empty
	\def\@oddfoot{\reset@font\hfil\thepage\hfil}%
	\let\@evenfoot\@oddfoot}

\begin{document}
	\begin{frontmatter}
		\title{Adaptive Algorithm for Sparse Signal Recovery}
		\author[]{Fekadu L. Bayisa$^{*}$}
		\author[]{Zhiyong Zhou}
		\author[]{Ottmar Cronie}
		\author[]{and Jun Yu}
		\address{Department of Mathematics and Mathematical Statistics, Ume{\aa} University, Ume{\aa}, Sweden}
		\cortext[author] {Corresponding author.\\\textit{E-mail address:} fekadu.bayisa@umu.se}
		\begin{abstract}
			The development of compressive sensing in recent years has  given much attention to sparse signal recovery. In sparse signal recovery, spike and slab priors are playing a key role in inducing sparsity. The use of such priors, however, results in non-convex and mixed integer programming problems. Most of the existing algorithms to solve non-convex and mixed integer programming problems involve either  simplifying assumptions, relaxations or high computational expenses. In this paper, we propose a new adaptive alternating direction method of multipliers (AADMM) algorithm to directly solve the suggested non-convex and mixed integer programming problem. The algorithm is based on the one-to-one mapping property of the support and non-zero element of the signal. At each step of the algorithm, we update the support by either adding an index to it or removing an index from it and use the alternating direction method of multipliers to recover the signal corresponding to the updated support. Moreover, as opposed to the competing "adaptive sparsity matching pursuit" and "alternating direction method of multipliers" methods our algorithm can solve non-convex problems directly.  Experiments on synthetic data and real-world images demonstrated that the proposed AADMM algorithm provides superior performance and is computationally cheaper than the recently developed iterative convex refinement (ICR) and adaptive matching pursuit (AMP) algorithms.
		\end{abstract}
	\end{frontmatter}
	\pagenumbering{arabic}
	\section{Introduction}
	Over the past decades, sparse signal recovery from  fewer samples has received attention due to increasing impracticality.
	Sparsity is usually assumed for signal reconstruction and inverse problems \citep{TroppJA}. It is often assumed in compressed sensing theory \citep{DonohoDL}. The presence of sparsity  has several applications in signal recovery \citep{TroppJA, WrightSJ}, medical image reconstruction \citep{Chaari1, AndersenMR}, image classification  \citep{MousavSrinivas, SrinivasU}, and dictionary learning \citep{SadeghiM, SuoYandDaoM}.

	Lately, there has been a continuous increase in the volume of  data generated. For instance, medical imaging systems such as magnetic resonance imaging (MRI) may deliver multidimensional signals. According to \cite{LeeSHLeeYH}, however, magnetic resonance (MR) images may appear blurred due to the long scanning times, which in turn can cause patient inconvenience and unwarranted movements. One of the strategies to accelerate MR scanning is image reconstruction from reduced samples. The key feature in MR image reconstruction is the use of prior information about the signal through the compressibility or sparse representation under an appropriate sparsifying transform such as the wavelet transform and finite-differencing \citep{LustigMandDonohoDL}. Signal reconstruction by taking the sparsity into account is therefore of great interest.
	
	In general, sparse signal recovery  problems are ill-posed problems. However, the sparsity assumption makes sparse signal recovery possible from fewer measurements. Hence, regularizations are often required to promote the sparsity of the unknown signal. Several algorithms have been proposed to solve these regularized problems. The most recent algorithms are  adaptive matching pursuit \citep{VuTH} and iterative convex refinement \citep{MousaviHS}. Greedy algorithms \citep{TroppJA, MohimaniHand, MousaviAandMaleki}, Bayesian  methods \citep{JiSandXue2, DobigeonNandHero, LuX1}, and general sparse approximation algorithms such as SpaRSA and alternating direction method of multipliers  \citep{WrightSJ, BeckerSandJBobin, BoydSandParikh} have also been exploited for sparse signal recovery.

	In Bayesian sparse signal recovery, setting priors for the signal has played a key role in promoting sparsity and improving performance. Examples of priors are spike and slab \citep{Mitchell, LuX2, AndersenMR, MousavSrinivas, MousaviHS, VuTH}, Bernoulli-Gaussian \citep{LavielleM}, Bernoulli-exponential \citep{DobigeonNandHero}, Laplacian \citep{BabacanSandMolina}, and generalized Pareto \citep{CevherV2}. 
	
	\cite{ChaariTourneretBatatia} explored a fully Bayesian sparse signal recovery using Bernoulli-Laplacian priors and obtained sparser solutions in comparison to the Bernoulli-Gaussian prior setting.  The increased sparsity is due to the Laplacian term in the model. Motivated by this,  we concentrated on the setup of \cite{YenT} to promote sparsity by setting spike and slab priors for the signal. We proposed a Bernoulli-Laplace prior for the signal in order to induce sparsity in the signal recovery. Using the proposed prior, we develop a model that is an extension of least absolute shrinkage and selection operator (LASSO) \citep{TibshiraniR}. The developed model is a more general model and its optimization is known to be a non-convex and mixed integer programming problem where the existing solving methods involve simplifying or relaxation assumptions \citep{AndersenMR, YenT, SrinivasU, MousaviHS}. In this work, however,  we developed an adaptive algorithm to solve the optimization problem directly in its general form. 
	
	To the best of our knowledge, the competitors of our algorithm are adaptive matching pursuit (AMP) \citep{VuTH} and iterative convex refinement (ICR) \citep{MousaviHS}. Our algorithm and AMP are similar in that they involve two stages to solve the the optimization problems. The first stage deals with the support of the signal while the second stage is concerned with signal reconstruction for a given support of the signal. We perform these two stages iteratively to obtain the reconstructed signal and an updated support of the signal. These two stages hold for AMP and AADMM. The main difference is the involvement of $l_{1}$-norm in the second stage. In AMP, the second stage does not involve $l_{1}$-norm  and one can reconstruct the signal by forward and backward substitution. While AADMM involves $l_{1}$-norm in the second stage with a cost of that the problem can not be solved by forward and backward substitution. However, it is advantageous for a reconstruction problem to involve an $l_{1}$-norm as it enforces sparsity during sparse signal recovery. In contrast to AADMM, ICR uses a relaxation assumption to solve the non-convex and mixed integer programming problem. 
	
	The main contributions of our work are the following:
	\begin{enumerate*}[label={\arabic*)}]
		\item We formulate a sparse model using the maximum a posteriori estimation technique.
		\item We propose an adaptive alternating direction method of multipliers, hereinafter AADMM, to solve the non-convex and mixed integer programming problem  directly. A matching pursuit procedure and an alternating direction method of multipliers are combined to develop a computationally efficient algorithm to solve the optimization problem. \cite{wu2012adaptive}  developed an adaptive sparsity matching pursuit algorithm,  which used two nested stages to solve convex problems, for sparse signal reconstruction. \cite{BoydSandParikh} also devised an alternating direction method of multipliers to solve convex optimization problems. Our algorithm can solve non-convex problems directly, which is a clear added-value that the above-mentioned two algorithms don't have. \item We compare the performance of the proposed algorithm with the recently developed algorithms. For a given support of the signal, the proposed optimization problem involves an $l_{1}$-norm. Consequently, the proposed optimization problem is not differentiable. As a result, we can not exploit Cholesky decomposition based forward-backward substitution to solve the optimization problem \citep{ BoydSandParikh, abur1988parallel}.  In the most recent algorithm AMP, the optimization problem is differentiable for a given support of the signal and its differentiation is available in a simple closed-form to exploit Cholesky decomposition based forward-backward substitution for solving the resulting closed form problem \citep{VuTH}. Therefore, the most recent algorithm AMP can not be used to solve our proposed optimization problem. In this regard, the proposed algorithm is advantageous. Besides, the proposed optimization problem involves an $l_{1}$-norm and therefore we expect the signal to be recovered in a sparser form. Here we compare the performance of the proposed algorithm with the recently developed algorithms AMP and ICR. Notice that ICR can be used to solve the proposed optimization while AMP can not be exploited to solve the optimization problem. As signal reconstruction is concerned, however, we can still use AMP to reconstruct the signal generated according to the problem setting in this article in order to compare its performance with AADMM.   \item The developed adaptive algorithm can also be used to reconstruct both unconstrained and constrained (or non-negative) signals. \item We test our algorithm on both simulated data and real images. The results reveal the merits of the proposed AADMM algorithm. 
	\end{enumerate*} 
	
	The paper is organized as follows. In Section \ref{sect2}, we demonstrate the details of the proposed model. The developed adaptive algorithm for sparse signal recovery, evaluation method of the signal recovery, and  the results obtained are reported in Section \ref{AADMM}, Section \ref{ESR}, and  Section \ref{Resu}, respectively.  Finally, the conclusions and the future works are presented in Section \ref{conc}.
	\section{Problem Formulation}\label{sect2}
	Sparse signal recovery algorithms are used to recover a sparse signal $\mathbf{x}\in{\Bbb R}^{n\times 1}$ from observed measurements $\mathbf{y}\in{\Bbb R}^{m\times 1}$, where $m\ll n$. The basic model for sparse signal recovery is given by
	\begin{eqnarray}\label{ABOq}
	\mathbf{y}=\mathbf{A}\mathbf{x} + \boldsymbol\epsilon, 
	\end{eqnarray}
	where $\mathbf{A}\in{\Bbb R}^{m\times n}$ is a measurement matrix, and $\boldsymbol\epsilon\in{\Bbb R}^{m\times 1}$ is a Gaussian noise with a variance-covariance structure given by $\sigma^2\mathbf{I}$.  Here $\mathbf{I}$ is an $m\times m$ identity matrix. Since $m\ll n$, the inverse problem in equation \eqref{ABOq}, which occurs in a number of applications in the field of signal and image processing \citep{Chaari1, Pustelnik} is an ill-posed problem. 
	
	To allow sparse modeling and regularization, we assume a prior distribution for the unknown signal $\mathbf{x}=\left(x_{1}, x _{2}, \cdots, x _{n} \right)^{T}$. Using a Bernoulli random variable $\omega_{i}$, assume that
	\begin{eqnarray*}
		x_{i}\mid \omega_{i} = 1\sim  P_{1}\left(x_{i}\right),\quad x_{i}\mid w_{i}=0 \sim P_{0}\left( x_{i}\right) \quad\text{and}\quad P_{0}\left(0\right) = 1,
	\end{eqnarray*}
	where  $P_{1}\left(\cdot\right)$ and $P_{0}\left(\cdot\right)$ are probability distributions and $i = 1, 2, \cdots, n$. Here $\omega_{i}$ is used to control the structural sparsity of the signal $\mathbf{x}$. It is also exploited to form a mixture of distributions given by 
	\begin{eqnarray}\label{Og2}
	x_{i} \sim  \omega_{i} P_{1}\left(x_{i}\right) +\left(1 - \omega_{i}\right)P_{0}\left(x_{i}\right).
	\end{eqnarray}
	The mixture of distributions in equation \eqref{Og2} is an approximation of a spike and slab prior, which was proposed by  \cite{Mitchell}. The distribution $P_{1}\left(x_{i}\right)$ can be thought of as a slab while $P_{0}\left(x_{i}\right)$, which is an approximation of the Dirac delta function centered at the event $x_{i} = 0$, can be considered as a spike. 
	Since the signal $\mathbf{x}$ is expected to be sparse,  one can specify a Laplace distribution peaked at location parameter $\mu = 0$ as a slab. In Bayesian inference, spike and slab priors are the gold standard for inducing sparsity \citep{Titsiasx}.  Our interest here is to explore sparse signal recovery using the following three-level models:   
	\begin{eqnarray}\label{Abiy2018}
	\mathbf{y}\mid \mathbf{A}, \mathbf{x}, \sigma^{2} &\sim& \mathcal{N}\left(\cdot\mid \mathbf{A}\mathbf{x}, \sigma^{2}\mathbf{I}\right),\notag\\
	\mathbf{x}\mid \boldsymbol\omega, \sigma^{2},\lambda&\sim& \prod_{i=1}^{n}\left\lbrace \omega_{i}\mathbf{Laplace}\left(\cdot\mid 0, \frac{2\sigma^{2}}{\lambda}\right) + \left(1 - \omega_{i}\right)P_{0}\left(x_{i}\right)\right\rbrace ,\\
	\boldsymbol\omega\mid \boldsymbol\kappa&\sim&\prod_{i=1}^{n}\mathbf{Bernoulli}\left(\kappa_{i}\right),\notag
	\end{eqnarray}
	\noindent where  
	\begin{eqnarray*}
		\boldsymbol\omega = \left(\omega_{1}, \omega_{2}, \cdots, \omega_{n}\right)',\hspace{0.02in}
		\boldsymbol\kappa = \left(\kappa _{1}, \kappa _{2}, \cdots, \kappa _{n}\right)', \hspace{0.02in}\lambda>0,
	\end{eqnarray*}
	and the notations $\mathcal{N}\left(\cdot\right)$, $\mathbf{Laplace}\left(\cdot\right)$ and $\mathbf{Bernoulli}\left(\cdot\right)$ represent normal, Laplace and Bernoulli distributions, respectively. 
	
	A posterior maximization procedure can be used to induce sparsity in the Bayesian frame work \citep{CevherV1, CevherV2, MousaviHS, VuTH}.
	\begin{proposition}\label{DchdQwr} Assume that  $\sigma^{2}$ is  known. Using the model set-up in equation \eqref{Abiy2018} and a posterior maximization procedure, one can obtain the regularized optimization problem:
		\begin{eqnarray}\label{ABO5}
		\min\limits_{\mathbf{x},\boldsymbol\omega}\left\lbrace ||\mathbf{y}-\mathbf{A}\mathbf{x}||^{2}_{2} +\lambda||\mathbf{x}||_{1}+\sum_{i=1}^{n}\omega_{i}\gamma_{i}\right\rbrace,
		\end{eqnarray}
		where
		\begin{eqnarray}\label{ABO243}
		\gamma_{i} = 2\sigma^{2}\log \frac{4\sigma^{2}\left(1-\kappa_{i}\right)}{\lambda \kappa_{i}}.
		\end{eqnarray}
	\end{proposition}
	\begin{proof}
		Let $\Omega= \left\lbrace \sigma^{2}, \lambda,\boldsymbol\kappa \right\rbrace$.  Using the hierarchical model, the joint posterior density of $\mathbf{x}$ and $\boldsymbol\omega$ is given by  
		\begin{eqnarray*}
			\log f\left(\mathbf{x}, \boldsymbol\omega\mid \mathbf{y}, \mathbf{A}, \Omega \right) &\propto&   \log\left\lbrace g\left(\mathbf{y}\mid\mathbf{A}, \mathbf{x}, \sigma^{2}\right)h\left(\mathbf{x}\mid\boldsymbol\omega, \sigma^{2}, \lambda\right)P\left(\boldsymbol\omega\mid \boldsymbol\kappa\right)\right\rbrace ,\\	
			&=&\log\Bigg\lbrace \mathcal{N}\left(\cdot\mid \mathbf{A}\mathbf{x}, \sigma^{2}\mathbf{I}\right)\Bigg\lbrace\prod_{i=1}^{n} \mathbf{Laplace}\left(\cdot\mid 0, \frac{2\sigma^{2}}{\lambda}\right)^{\omega_{i}}\times\\&&P_{0}\left( x_{i}\right)^{\left(1 - \omega_{i}\right)}\Bigg\rbrace \prod_{i=1}^{n}\mathbf{Bernoulli}\left(\kappa_{i}\right)\Bigg\rbrace,\\
			&=&-\frac{n}{2}\log 2\pi - \frac{n}{2}\log\sigma^{2}-\frac{1}{2\sigma^{2}}||\mathbf{y}-\mathbf{A}\mathbf{x}||^{2}_{2}+\\&&\sum_{i=1}^{n}\Bigg\lbrace-\omega_{i}\log\frac{4\sigma^{2}}{\lambda}-\omega_{i}\log\left(1-\kappa_{i}\right)+\omega_{i}\log \kappa_{i}\Bigg\rbrace+\\&&\sum_{i=1}^{n}\log\left(1-\kappa_{i}\right)-\frac{\lambda}{2\sigma^{2}}\sum_{i=1}^{n}\omega_{i}|x_{i}|.
		\end{eqnarray*}
		Assume that $\sigma^{2}$ is known. Then, we have the following optimization problem: 
		\begin{align*}
		\max\limits_{\mathbf{x}, \boldsymbol\omega}\left\lbrace 2\sigma^{2}\log f\left(\mathbf{x}, \boldsymbol\omega\mid \mathbf{y}, \mathbf{A}, \Omega \right)\right\rbrace=&\min\limits_{\mathbf{x},\boldsymbol\omega}\left\lbrace -2\sigma^{2}\log f\left(\mathbf{x}, \boldsymbol\omega\mid \mathbf{y}, \mathbf{A}, \Omega\right)\right\rbrace,\\ =&  \min\limits_{\mathbf{x},\boldsymbol\omega}\Bigg\lbrace n\sigma^{2}\log 2\pi + n\sigma^{2}\log\sigma^{2}+||\mathbf{y}-\mathbf{A}\mathbf{x}||^{2}_{2}+\notag\\&\lambda\sum_{i=1}^{n}\omega_{i}|x_{i}|-2\sigma^{2}\sum_{i=1}^{n}\log\left(1-\kappa_{i}\right)-2\sigma^{2}\times\\&\sum_{i=1}^{n}\Bigg\lbrace\omega_{i}\log\frac{4\sigma^{2}}{\lambda}+\omega_{i}\log\left(1-\kappa_{i}\right)-\omega_{i}\log \kappa_{i}\Bigg\rbrace\Bigg\rbrace,\\
		=&\min\limits_{\mathbf{x},\boldsymbol\omega}\left\lbrace ||\mathbf{y}-\mathbf{A}\mathbf{x}||^{2}_{2} +\lambda||\mathbf{x}||_{1}+\sum_{i=1}^{n}\omega_{i}\gamma_{i}\right\rbrace.
		\end{align*}
	\end{proof}
	\begin{remark}
		\upshape
		If all $\omega_{i}$ are identically distributed, then the last term of the optimization problem in equation \eqref{ABO5} becomes a regularized  $l_{0}$-norm of the signal. From equation \eqref{ABO243}, $\gamma_{i}$ can be negative for large $\kappa_{i}$ and increasing $\kappa_{i}$ decreases $\gamma_{i}$. This means that a strong belief in the presence of a non-zero value of the signal decreases the penalty value for the signal.
	\end{remark}
	\begin{remark}
		\upshape
		We have proposed a more general optimization problem in comparison to the optimization problem suggested by \cite{YenT}, \cite{LuX2}, \cite{LuX1}, and \cite{SrinivasU} since they simplified the optimization by assuming the same parameter $\kappa$ for all the Bernoulli random variables. This assumption allowed them to obtain an optimization problem that involves a regularized $l_{0}$-norm. In addition to the assumption of having the same parameter for all the Bernoulli random variables, \cite{YenT} exploited the $l_{2}$-norm of the signal instead of the $l_{1}$-norm in equation \eqref{ABO5}. The same problem setting as in equation \eqref{ABO5} has been utilised by \cite{MousaviHS} and \cite{VuTH} except that they exploited the $l_{2}$-norm  instead of the $l_{1}$-norm. This makes our approach advantageous as it contains an $l_{1}$-norm, which can enforce sparse signal recovery \citep{MalioutovD, JiaXandMen}. The problem in equation \eqref{ABO5} promotes  greater generality in capturing sparsity and it is also an extension of the existing LASSO method. Our optimization problem appears in signal recovery \citep{LuX1,LuX2, Chaari2}, regression \citep{TibshiraniR}, image classification \citep{SrinivasU}, and medical image reconstruction \citep{LustigMandDonohoDL}. Using the conventional optimization algorithms, we may not be able to solve the proposed optimization problem because it is a non-convex mixed-integer programming problem.
	\end{remark}
	\section{Adaptive Algorithm}\label{AADMM}
	In this section, we present the details of the methods used to solve the proposed optimization problem.  Let $\textbf{a}_{i}$ denote the $i^{th}$ column of the matrix $\mathbf{A}$. Each column of the matrix is assumed to have norm 1. That is, $||\textbf{a}_{i}||^{2}_{2} =1$, $i= 1, 2, \cdots, n$.  If we know the support $S= \left\lbrace  i \mid x_{i} \neq 0, i= 1, 2, \cdots, n\right\rbrace$  of the signal $\mathbf{x}$, then the optimization problem in equation \eqref{ABO5} is equivalent to:
	\begin{eqnarray}\label{ABO52}
	\min\limits_{\mathbf{x}_{S}}\left\lbrace ||\mathbf{y}-\mathbf{A}^{S}\mathbf{x}_{S}||^{2}_{2} +\lambda||\mathbf{x}_{S}||_{1} + \sum_{i\in S}\gamma_{i}\right\rbrace, 
	\end{eqnarray}
	\noindent where 
	\begin{eqnarray*}
		\mathbf{A}^{S}  = \left[\textbf{a}_{i}: i\in S\right],\quad
		\mathbf{x}_{S} = \left(x_{i}: i\in S\right)^{T}. 
	\end{eqnarray*}
	Since there is a one-to-one correspondence between $S$  and $\mathbf{x}_{S}$,  we may solve the optimization problem in equation \eqref{ABO5} by finding  the support $S$ and solving the optimization problem in equation \eqref{ABO52}. Based on this idea, an adaptive alternating direction method of multipliers (AADMM) is proposed for the signal recovery. AADMM uses a greedy method to update the support $S$. For a given support, it uses alternating direction method of multipliers (ADMM) to solve the optimization problem in equation \eqref{ABO52}. For each iteration of the algorithm, the greedy method updates the support $S$ either by absorbing one of the unselected indices into $S$ or by removing one of the elements from $S$. We select the option, that is either absorbing or removing an index, that decreases the objective  function of the optimization problem in equation \eqref{ABO5}.
	
	Based on the support $S$ of the signal, consider the following definitions:
	\begin{eqnarray}\label{ABO55}
	\delta_{S} = \sum_{i\in S}\gamma_{i} \hspace{0.06in}\text{and}\hspace{0.06in} g\left(S\right)   = \min\limits_{\mathbf{x}_{S}}\left\lbrace ||\mathbf{y}-\mathbf{A}^{S}\mathbf{x}_{S}||^{2}_{2} +\lambda||\mathbf{x}_{S}||_{1} + \delta_{S}\right\rbrace. 
	\end{eqnarray}
	For each iteration in the AADMM algorithm, we need to compute the following to update the support: 
	\begin{eqnarray}
	U_{S}  &=& \min\limits_{i \notin S}\left\lbrace  g\left(S\cup\left\lbrace i\right\rbrace \right)-g\left(S\right)\right\rbrace,\label{ABO53}\\ 
	i^{*}&=& \arg\min\limits_{\hspace{-0.2in}i \notin S}\left\lbrace  g\left(S\cup\left\lbrace i\right\rbrace \right)-g\left(S\right)\right\rbrace,\notag\\ 
	V_{S}  &=& \min\limits_{j \in S}\left\lbrace  g\left(S\setminus\left\lbrace j\right\rbrace \right)-g\left(S\right)\right\rbrace,\label{ABO15}\\
	j^{*}&=& \arg\min\limits_{\hspace{-0.2in}j \in S}\left\lbrace  g\left(S\setminus\left\lbrace j\right\rbrace \right)-g\left(S\right)\right\rbrace.\notag
	\end{eqnarray}
	Equations \eqref{ABO53} and \eqref{ABO15} represent the minimization of the change in the cost function by selecting one of
	unselected indices and  by removing one of the already selected indices, respectively. Based on $U_{S}$ and $V_{S}$, we have three possible cases. The first case is that if both $U_{S}$ and $V_{S}$ are not less than zero, adding or removing the indices can not decrease the cost function and we can stop the algorithm. If $U_{S} < V_{S}$, then  we update $S$ by absorbing the index $i^{*}$ and if $V_{S}<U_{S}$, then we update $S$ by removing the index $j^{*}$.  After each iteration,  the procedure guarantees that the cost function decreases and  thereby the algorithm converges eventually, being a monotone limit. The optimal support $\hat{S}$ is an element of the set $\left\lbrace S:  U_{S}\geq0, \quad V_{S}\geq0\right\rbrace.$  However, it is hardly practical to optimize $g\left(S\cup\left\lbrace i\right\rbrace \right)$ and $g\left(S\setminus\left\lbrace j\right\rbrace \right)$ for each $i$ and $j$. Therefore, we would like to use the upper bounds of $U_{S}$ and $V_{S}$ to significantly reduce the computational cost of $g\left(S\cup\left\lbrace i\right\rbrace \right)$ and $g\left(S\setminus\left\lbrace j\right\rbrace \right)$. For each iteration in the AADMM algorithm, the decision to add an index to $S$ or remove an index from $S$, that is to obtain an updated support, is based on the upper bounds of $U_{S}$ and $V_{S}$. Let $S^{u}$ denote the updated support for a given iteration in the AADMM algorithm. 
	Using the updated support $S^{u}$, we estimate $\mathbf{x}_{S^{u}}$ and compute the residual $\mathbf{r}_{S^{u}} =\mathbf{y}-\mathbf{A}^{S^{u}}\mathbf{x}_{S^{u}}$. This continues until convergence of the algorithm has been obtained. The following results are the key tools for developing the algorithm. Proposition \ref{Qeerro1} is used to initialise the support while propositions \ref{Qeerro2} and \ref{Qeerro3} are utilised for approximating the computationally expensive  $U_{S}$ and $V_{S}$. Once the initial value of the support is obtained, we estimate the signal using ADMM \citep{BoydSandParikh}. 
	\begin{proposition} \label{Qeerro1}
		If $\gamma_{i}<0$, then $i \in \hat{S}$.
	\end{proposition}
	\begin{proof}
		Assume that $i \notin \hat{S}$. Using equation \eqref{ABO55}, we have 
		\begin{eqnarray*}
			g\left(\hat{S}\cup\left\lbrace i\right\rbrace \right) \leq ||\mathbf{r}_{\hat{S}}-x_{i}\mathbf{a}_{i}||^{2}_{2} +\lambda||\mathbf{x}_{\hat{S}}||_{1} + \delta_{\hat{S}} + \lambda|x_{i}| + \gamma_{i}
			= g\left(\hat{S}\right) + x_{i}^{2} + \lambda|x_{i}| -2x_{i}\mathbf{r}^{T}_{\hat{S}}\mathbf{a}_{i} + \gamma_{i},
		\end{eqnarray*}
		which implies that 
		\begin{eqnarray*}
			g\left(\hat{S}\cup\left\lbrace i\right\rbrace \right) - g\left(\hat{S}\right)\leq x_{i}^{2} + \lambda|x_{i}| -2x_{i}\mathbf{r}^{T}_{\hat{S}}\mathbf{a}_{i} + \gamma_{i}.
		\end{eqnarray*}
		Let $h\left( x\right) =x^{2} + \lambda|x| -2x\mathbf{r}^{T}_{\hat{S}}\mathbf{a}_{i} + \gamma_{i}$.  We see that $h\left( x\right)$ is continuous, $\lim_{x\rightarrow\infty}h\left( x\right)= \infty$ and $h\left( 0\right) =  \gamma_{i}< 0$.  We observe that there exists a value $\bar{x}$ of $x_{i}$ such that $\gamma_{i}<h\left( \bar{x}\right)< 0$, whereby we have that $g\left(\hat{S}\cup\left\lbrace i\right\rbrace \right) - g\left(\hat{S}\right)< h\left( \bar{x}\right)< 0$. This implies that $\hat{S}$ can not be the optimal solution, which is a contradiction to the assumption. Therefore, $i$ must be in $\hat{S}$.
	\end{proof}
	\begin{proposition}\label{Qeerro2}
		$U_{S}$ given in equation \eqref{ABO53} satisfies
		\begin{eqnarray} 
		U_{S} \leq   
		\bar{U}_{S} = \min\limits_{i \notin S}\begin{cases}
		\gamma_{i}-\left(\mathbf{r}^{T}_{S}\mathbf{a}_{i}\right)^{2} +\lambda\mathbf{r}^{T}_{S}\mathbf{a}_{i}
		-\frac{\lambda^{2}}{4}, & \text{if} \hspace{0.1in} \mathbf{r}^{T}_{S}\mathbf{a}_{i}>\frac{\lambda}{2}, \\
		\gamma_{i}, & \text{if} \hspace{0.1in} |\mathbf{r}^{T}_{S}\mathbf{a}_{i}|\leq\frac{\lambda}{2},\label{ABO56q}\\
		\gamma_{i}-\left(\mathbf{r}^{T}_{S}\mathbf{a}_{i}\right)^{2} -\lambda\mathbf{r}^{T}_{S}\mathbf{a}_{i}
		-\frac{\lambda^{2}}{4}, & \text{if} \hspace{0.1in} \mathbf{r}^{T}_{S}\mathbf{a}_{i}<-\frac{\lambda}{2},
		\end{cases}
		\end{eqnarray}
	\end{proposition}
	\begin{proof}
		Since
		\begin{eqnarray*}
			g\left(S\cup\left\lbrace i\right\rbrace \right) - g\left(S\right)\leq x_{i}^{2} + \lambda|x_{i}| -2x_{i}\mathbf{r}^{T}_{S}\mathbf{a}_{i} + \gamma_{i},
		\end{eqnarray*}
		for every $x_{i}$, we have that 
		\begin{eqnarray*}
			g\left(S\cup\left\lbrace i\right\rbrace \right) - g\left(S\right)\leq \min\limits_{x_{i}}\left\lbrace x_{i}^{2} + \lambda|x_{i}| -2x_{i}\mathbf{r}^{T}_{S}\mathbf{a}_{i} + \gamma_{i}\right\rbrace .
		\end{eqnarray*}
		This implies that 
		\begin{eqnarray}\label{Qeerro92}
		\min\limits_{i \notin S}\left\lbrace g\left(S\cup\left\lbrace i\right\rbrace \right) - g\left(S\right)\right\rbrace \leq  \min\limits_{i \notin S}\left\lbrace \min\limits_{x_{i}}\left\lbrace x_{i}^{2} + \lambda|x_{i}| -2x_{i}\mathbf{r}^{T}_{S}\mathbf{a}_{i} + \gamma_{i}\right\rbrace\right\rbrace.
		\end{eqnarray}
		We obtain the result in equation \eqref{ABO56q} from equation \eqref{Qeerro92}.
	\end{proof}
	\begin{proposition}\label{Qeerro3}
		For a support $S$,  $V_{S}$ given in equation \eqref{ABO15} satisfies 
		\begin{eqnarray}\label{ABO58} 
		V_{S} \leq \bar{V}_{S} = \min\limits_{j \in S}\left\lbrace x_{j}^{2} - \lambda|x_{j}| + 2x_{j}\mathbf{r}^{T}_{S}\mathbf{a}_{j} - \gamma_{j}\right\rbrace,
		\end{eqnarray}
		where $j\in S$ and $x_{j}$ is the element of the vector $\mathbf{x}_{S}$.
	\end{proposition}
	\begin{proof}
		For any $j\in S$,
		\begin{eqnarray}
		g\left(S\right) &=& ||\mathbf{r}_{S}||^{2}_{2} +\lambda||\mathbf{x}_{S}||_{1} + \delta_{S},\notag \\
		&=&||\mathbf{r}_{S}-x_{j}\mathbf{a}_{j}+x_{j}\mathbf{a}_{j}||^{2}_{2} +\lambda||\mathbf{x}_{S\setminus\left\lbrace j\right\rbrace }||_{1} + \delta_{S\setminus\left\lbrace j\right\rbrace }+ \lambda|x_{j}| + \gamma_{j},\notag\\
		&=&||\mathbf{r}_{S\setminus\left\lbrace j\right\rbrace }||^{2}_{2} +\lambda||\mathbf{x}_{S\setminus\left\lbrace j\right\rbrace }||_{1} + \delta_{S\setminus\left\lbrace j\right\rbrace }-x_{j}^{2}-2x_{j}\mathbf{a}^{T}_{j}\mathbf{r}_{S} + \lambda|x_{j}| + \gamma_{j},\notag\\
		&\ge& g\left({S\setminus\left\lbrace j\right\rbrace }\right)-x_{j}^{2}-2x_{j}\mathbf{a}^{T}_{j}\mathbf{r}_{S} + \lambda|x_{j}| + \gamma_{j}. \label{ABO57}
		\end{eqnarray}
		Equation \eqref{ABO57} implies that 
		\begin{eqnarray*}
			g\left({S\setminus\left\lbrace j\right\rbrace }\right)-g\left(S\right)\le x_{j}^{2} + 2x_{j}\mathbf{a}^{T}_{j}\mathbf{r}_{S} - \lambda|x_{j}|  - \gamma_{j},
		\end{eqnarray*}
		which holds for every $j\in S$ and we proved that equation \eqref{ABO58} holds.
	\end{proof}
	We exploit the computationally cheap $\bar{U}_{S}$ and $\bar{V}_{S}$  to update the support and we use the updated support to update the estimation of the signal. This procedure continues until we obtain the optimal support and signal. We have summarised the estimation procedure in Algorithm \ref{Oromia202}.
	\newcommand{\sfunction}[1]{\textsf{\textsc{#1}}}
	\algrenewcommand\algorithmicforall{\textbf{foreach}}
	\algrenewcommand\algorithmicindent{.8em}
	\begin{algorithm}[!h]
		\caption{\textbf{AADMM} Algorithm}\label{Oromia202}
		\begin{algorithmic}[1]	
			\State Inputs: $\mathbf{y}$, $\mathbf{A}$, $\lambda$, $\boldsymbol\gamma$,
			\State Initialise the support: $S=\left\lbrace i \mid \gamma_{i}<0\right\rbrace$,
			\While {true}
			\State Solve $\mathbf{x}_{S}$: \textbf{ADMM}, \label{AB7}
			\State Update the residual: $\mathbf{r}_{S} = \mathbf{y}-\mathbf{A}^{S}\mathbf{x}_{S}$,
			\State Compute $\left[\bar{U}_{S}, i^{*}\right]$ \label{AB6}, where 
			\begin{eqnarray}   
			\bar{U}_{S} = \min\limits_{i \notin S}\begin{cases}
			\gamma_{i}-\left(\mathbf{r}^{T}_{S}\mathbf{a}_{i}\right)^{2} +\lambda\mathbf{r}^{T}_{S}\mathbf{a}_{i}
			-\frac{\lambda^{2}}{4}, & \text{if} \hspace{0.1in} \mathbf{r}^{T}_{S}\mathbf{a}_{i}>\frac{\lambda}{2}, \\
			\gamma_{i}, & \text{if} \hspace{0.1in} |\mathbf{r}^{T}_{S}\mathbf{a}_{i}|\leq\frac{\lambda}{2},\label{ABO56906}\\
			\gamma_{i}-\left(\mathbf{r}^{T}_{S}\mathbf{a}_{i}\right)^{2} -\lambda\mathbf{r}^{T}_{S}\mathbf{a}_{i}
			-\frac{\lambda^{2}}{4}, & \text{if} \hspace{0.1in} \mathbf{r}^{T}_{S}\mathbf{a}_{i}<-\frac{\lambda}{2},
			\end{cases}
			\end{eqnarray}
			\begin{eqnarray*}  	
				i^{*}= \arg\min\limits_{\hspace{-0.2in}i \notin S}
				\begin{cases}
					\gamma_{i}-\left(\mathbf{r}^{T}_{S}\mathbf{a}_{i}\right)^{2} +\lambda\mathbf{r}^{T}_{S}\mathbf{a}_{i}
					-\frac{\lambda^{2}}{4}, & \text{if} \hspace{0.1in} \mathbf{r}^{T}_{S}\mathbf{a}_{i}>\frac{\lambda}{2}, \\
					\gamma_{i}, & \text{if} \hspace{0.1in} |\mathbf{r}^{T}_{S}\mathbf{a}_{i}|\leq\frac{\lambda}{2},\\
					\gamma_{i}-\left(\mathbf{r}^{T}_{S}\mathbf{a}_{i}\right)^{2} -\lambda\mathbf{r}^{T}_{S}\mathbf{a}_{i}
					-\frac{\lambda^{2}}{4}, & \text{if} \hspace{0.1in} \mathbf{r}^{T}_{S}\mathbf{a}_{i}<-\frac{\lambda}{2},
				\end{cases}
			\end{eqnarray*}
			\State Compute $\left[\bar{V}_{S}, j^{*}\right]$ , where 
			\begin{eqnarray*}
				\bar{V}_{S} = \min\limits_{j \in S}\left\lbrace x_{j}^{2} - \lambda|x_{j}| + 2x_{j}\mathbf{r}^{T}_{S}\mathbf{a}_{j} - \gamma_{j}\right\rbrace,
			\end{eqnarray*}
			\begin{eqnarray*}
				j^{*}= \arg\min\limits_{\hspace{-0.2in}j \in S}\left\lbrace x_{j}^{2} - \lambda|x_{j}| + 2x_{j}\mathbf{r}^{T}_{S}\mathbf{a}_{j} - \gamma_{j}\right\rbrace,
			\end{eqnarray*}
			\State Decide: 
			\If{$\min\left\lbrace \bar{U}_{S}, \bar{V}_{S}\right\rbrace \ge 0$,} break the while loop,
			\Else\If{$ \bar{U}_{S} < \bar{V}_{S}$,} 
			\State Insert index: $S = S\cup\left\lbrace i^{*} \right\rbrace$, 
			\Else
			\State Remove index: $S = S\setminus\left\lbrace j^{*} \right\rbrace$
			\EndIf
			\EndIf
			\EndWhile
			\State Outputs: $S \rightarrow \hat{\boldsymbol\omega}$ and $\mathbf{x}_{S} \rightarrow \hat{\mathbf{x}}$.
		\end{algorithmic}
	\end{algorithm}
	
	The optimization problem in equation \eqref{ABO5} with non-negative constraint on the sparse signal can also be solved by AADMM. For the non-negative constraint, we need to modify steps \ref{AB7} and \ref{AB6} of Algorithm \ref{Oromia202}. For step \ref{AB7}, we need to modify the ADMM algorithm provided by \cite{BoydSandParikh}. That is, we replace the soft-thresholding operator  $S_{\beta}\left( \Lambda \right)$ in the ADMM algorithm by
	\begin{eqnarray*}  
		\max\left\lbrace 0, S_{\beta}\left(\Lambda\right)\right\rbrace,
	\end{eqnarray*}
	where the soft-thresholding operator is defined by
	\begin{eqnarray*}  
		S_{\beta}\left(\Lambda\right)=
		\begin{cases}
			\Lambda-\beta, &  \Lambda>\beta ,\\
			0, & |\Lambda|\le \beta,\\
			\Lambda+\beta, &\Lambda<\beta.
		\end{cases}
	\end{eqnarray*}
	Furthermore, we need to substitute $\bar{U}_{S}$ of step \ref{AB6} (that is equation \eqref{ABO56906}) by
	\begin{eqnarray*}  
		\bar{U}_{S} =\min\limits_{i \notin S}\left\lbrace \left( \max\left\lbrace 0, \mathbf{r}^{T}_{S}\mathbf{a}_{i}-\frac{\lambda}{2}\right\rbrace \right)^{2} +\lambda\max\left\lbrace 0, \mathbf{r}^{T}_{S}\mathbf{a}_{i}-\frac{\lambda}{2}\right\rbrace-2\max\left\lbrace 0, \mathbf{r}^{T}_{S}\mathbf{a}_{i}-\frac{\lambda}{2}\right\rbrace \mathbf{r}^{T}_{S}\mathbf{a}_{i}+\gamma_{i}\right\rbrace.
	\end{eqnarray*}
	The sparse signal recovery performance of the ICR algorithm \citep{MousaviHS} have been compared with the algorithms: majorization-minimization algorithm \citep{YenT}, FOCUSS \citep{gorodnitsky1997sparse}, expectation propagation approach for spike and slab recovery \citep{hernandez2015expectation}, and Variational Garrote \citep{kappen2014variational}. Since the ICR algorithm performs better than these algorithms \citep{MousaviHS} and as it can also be used to solve our proposed optimization problem, we choose to compare the sparse signal recovery performance of our algorithm with the recent ICR algorithm. Although AMP can not be used to solve our optimization problem, we can still apply AMP to the signal generated according to our model setup to  compare its performance with the proposed AADMM as AMP outperforms ICR and other algorithms, see \cite{VuTH}.  
	
	\section{Evaluation of the signal recovery}\label{ESR}
	Mean squared error (MSE) is used to evaluate the closeness of the recovered signal to the ground truth. Besides, we utilise the support match level (SML)  to measure how much the support of the recovered signal matches that of the support of the ground truth. We also study the effectiveness of the AADMM algorithm in comparison to the most recent ICR and AMP algorithms using MSE, SML, computational time (CT), objective function value (OFV), and sparsity level (SL). Recall that we only compare our algorithm to the ICR and AMP algorithms since they have the best performance among the competing algorithms, see \cite{MousaviHS} and  \cite{VuTH}. 
	\section{Results}\label{Resu}
	In this section, we present numerical results of our sparse signal recovery algorithm. In comparison to the ICR and AMP algorithms, we also demonstrate the worthiness of the AADMM algorithm. 
	\subsection{Simulation results for unconstrained signals}
	Our simulation setup  for sparse signal recovery was as in \cite{YenT, BeckAandTeboulleM}. We have used  $\lambda=2\times 10^{-4}$, $\sigma^{2}= 3.24\times10^{-4}$, and randomly  generated a Laplacian sparse vector $\mathbf{x}\in{\Bbb R}^{512\times 1}$ with 30 non-zeros.  Using $\mathbf{x}$, an additive Gaussian noise with variance $\sigma^{2}$ and a randomly generated Gaussian matrix $\mathbf{A} \in{\Bbb R}^{128\times 512}$,  we obtained an observation vector $\mathbf{y}\in{\Bbb R}^{128\times 1}$ according to equation \eqref{ABOq}. Using 500 different trials for $\mathbf{A}$,  $\mathbf{x}$, and $\epsilon$, the averages of the evaluation results are presented in Table \ref{tab:AADMM1}. 
	\begin{table}[H]
		\normalsize\caption{Comparison of AADMM, AMP, and ICR algorithms for sparse signal recovery using a true sparsity level = 30. We exploited OFV, MSE, SML, SL, and CT to compare the performance of the algorithms.}
		\label{tab:AADMM1}
		\centering
		\begin{tabular}{c|ccccc}
			\hline \hline
			Method&OFV&MSE&SML(\%)&SL&CT(S)\\\hline
			AADMM& 0.058&1.44$\times 10^{-4}$&96.256& 27.394 &0.024\\
			AMP & 0.042&1.75$\times 10^{-4}$& 95.395&32.176&0.122\\
			ICR & 0.154&1.06$\times 10^{-3}$&92.950&36.066 &5.200\\\hline\hline
		\end{tabular}
	\end{table}
It can be seen from Table \ref{tab:AADMM1} that AADMM outperforms ICR in several aspects. In terms of OFV, we see that the signal recovery using AADMM achieved a lower OFV. The result in  the table also revealed that the AADMM solution is closer to the ground truth in terms of MSE. Besides, AADMM is computationally faster than ICR. We also utilised SML to measure how much the support of the recovered signal matches that of the support of the ground truth. It is clear from the table that AADMM provides sparser solution and a higher SML (96.256\%) than ICR. 
	
In comparison to AMP,  AADMM achieved a better result in recovering a signal closer to the ground truth. Moreover, AADMM is computationally faster and provided sparser solution. However, AMP achieved a lower OFV as compared to AADMM which may be due to AADMM's sparser signal recovery property.
		
Using different sparsity levels, we evaluated the algorithms and presented their corresponding evaluation results of the sparse signal recovery problem in Figure \ref{Or21}. 
	\begin{figure}[H]
		\centering
		\includegraphics[height=3cm, width=7cm]{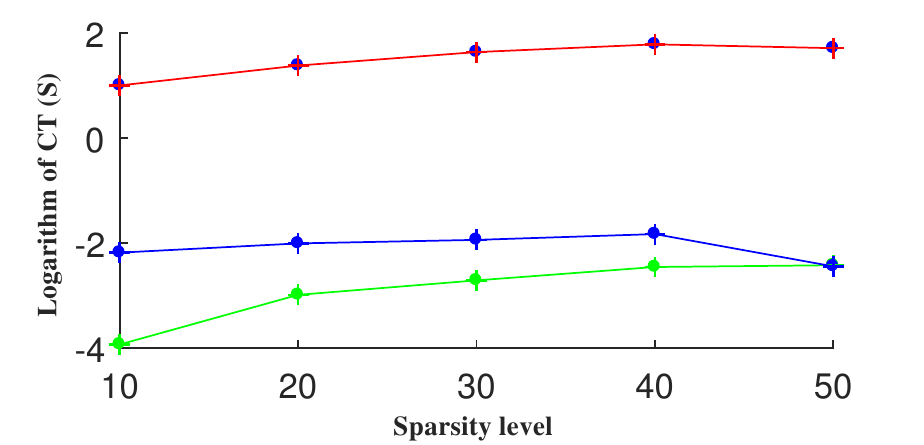}
		\includegraphics[height=3cm, width=7cm]{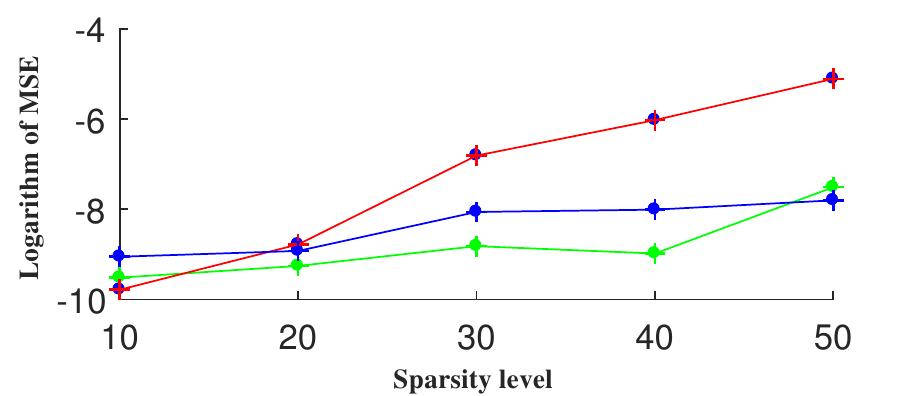}\\
		\includegraphics[height=3cm, width=7cm]{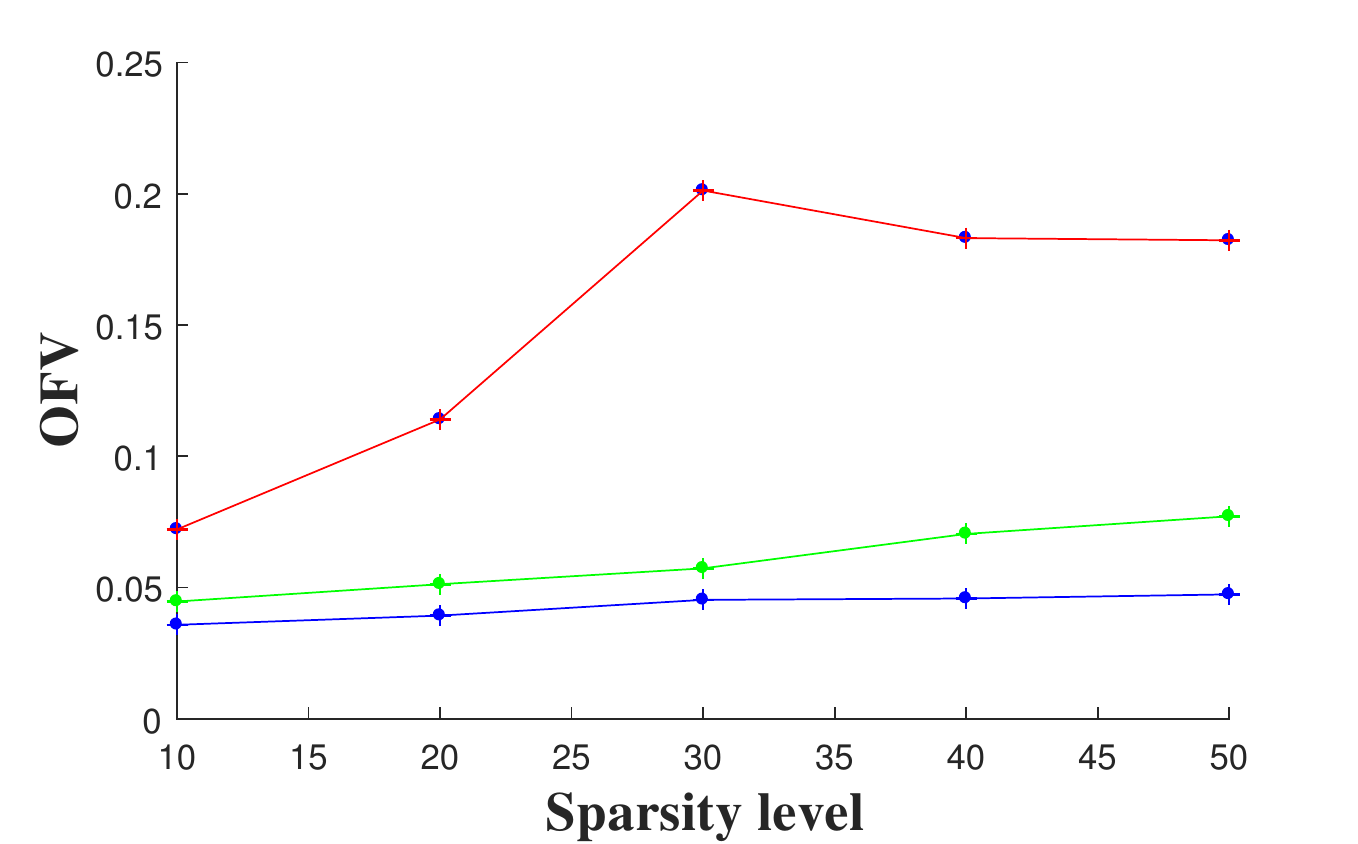}
		\includegraphics[height=3cm, width=7cm]{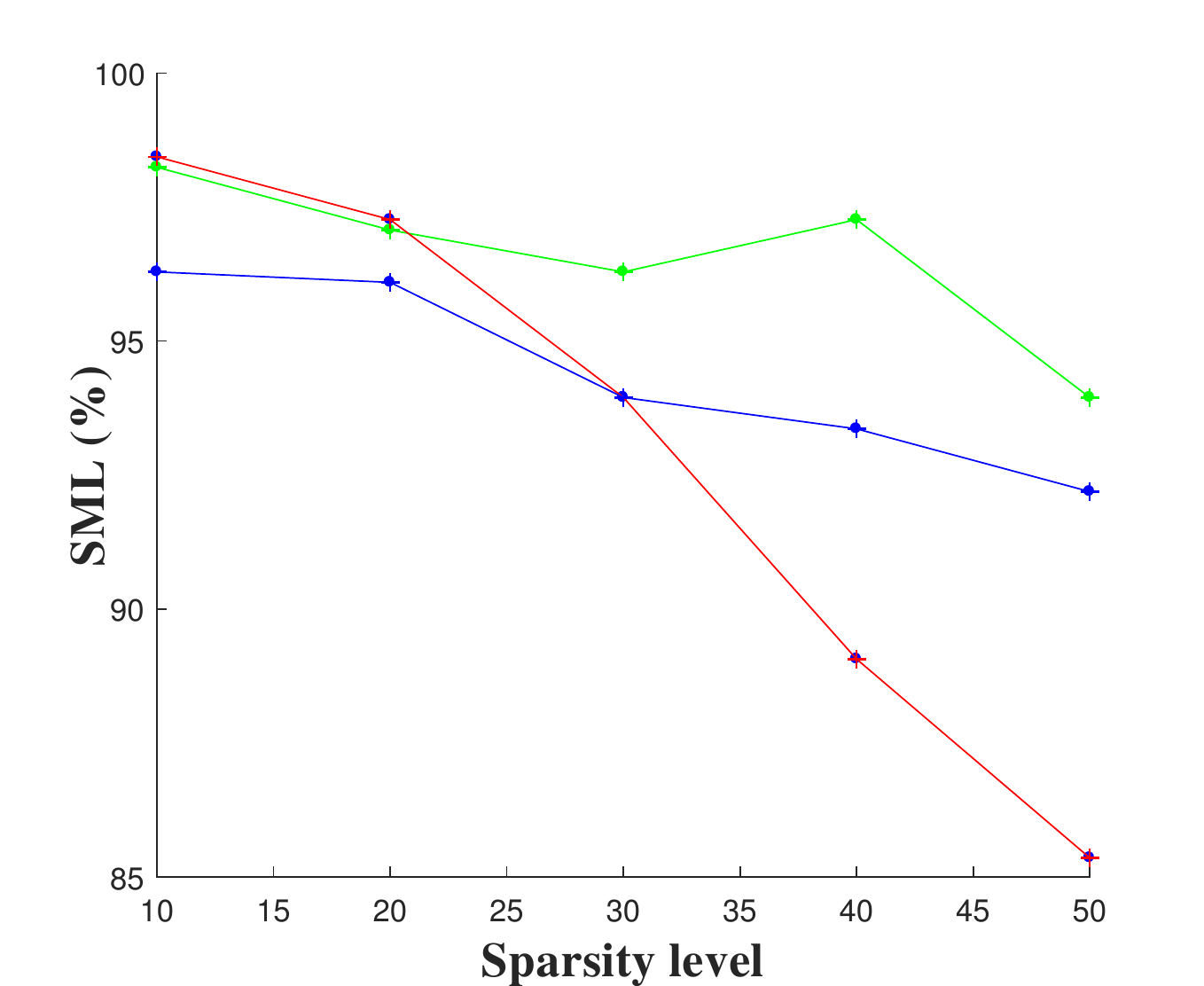}
		\normalsize\caption{Evaluation of the sparse signal recovery for different sparsity levels using AADMM (green), AMP (blue)  and ICR (red). CT, OFV, MSE, and SML are used to evaluate the performance of the algorithms over different sparsity levels. The computational times and MSEs are transformed by logarithm to make their plots visible.}
		\label{Or21}
	\end{figure}
	The figure shows that our method consistently outperforms ICR and AMP in terms of SML and MSE. It should be noted that AMP achieved a lower OFV in comparison to AADMM, which may be due to the sparser signal recovery property of AADMM and the smoothing term in the optimization problem of AMP. It is worth emphasising that AMP and ICR are also computationally more costly than AADMM.
	
	Figure \ref{Or223} presents the assessment of the sparse signal  recovery problem using different noise levels. We utilised  MSE and SML to evaluate the performance of the algorithms for different noise levels.
	\begin{figure}[H]
		\centering
		\includegraphics[height=4cm, width=7cm]{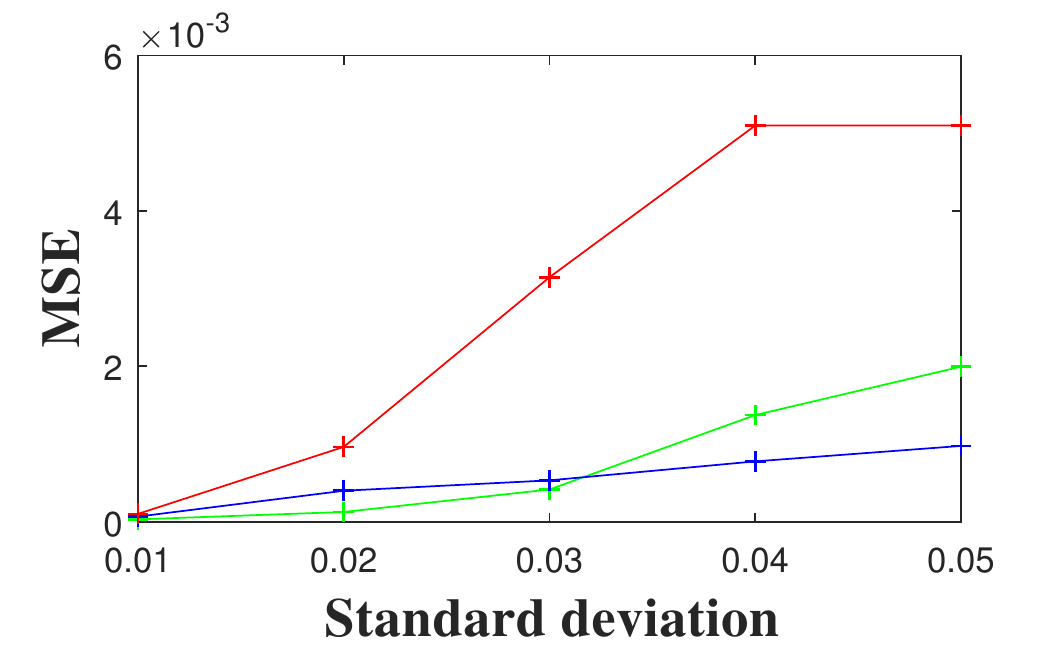}
		\includegraphics[height=4cm, width=7cm]{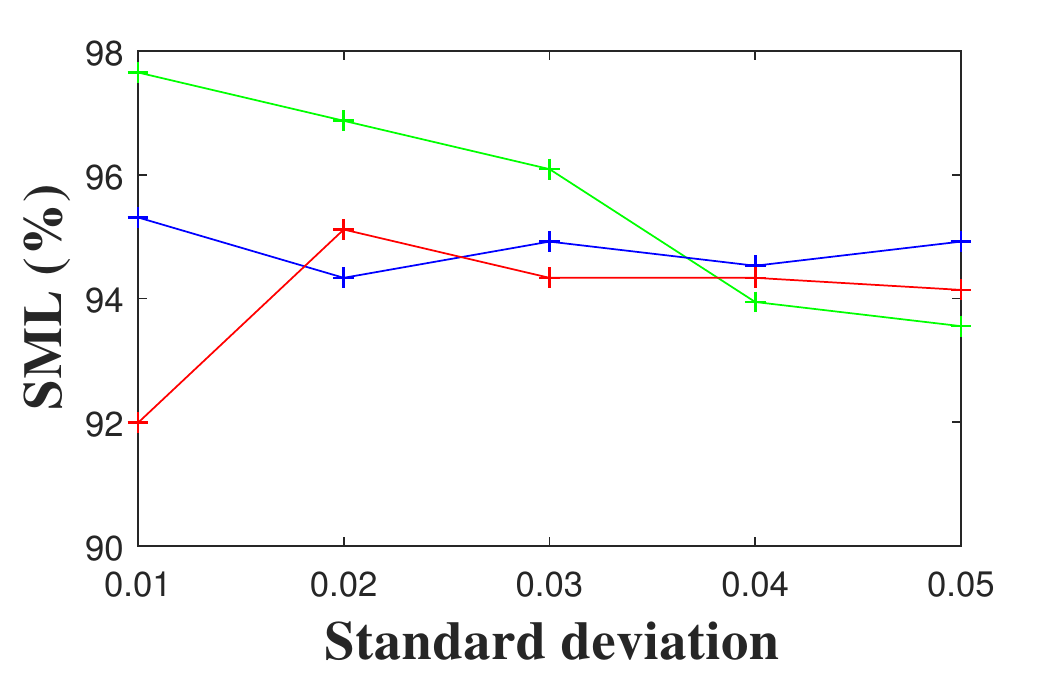}
		\normalsize\caption{Evaluation of AADMM (green), AMP (blue), and ICR (red) using different noise levels. MSE and SML are used to evaluate the performance of the algorithms over different noise levels.}
		\label{Or223}
	\end{figure}
	The figure shows that the recovered sparse signal using the AADMM algorithm is closer to the ground truth for lower noise levels. In terms of the SML,  AADMM  is more robust than AMP and ICR on  most of the noise levels. However, the use of high noise levels has a negative effect on the sparse signal recovery.
	
	Figure \ref{Or223345}  demonstrates full histograms of the mean squared errors (MSEs) and standard errors (SEs) for the AADMM, AMP, and ICR algorithms.
	\begin{figure}[H]
		\centering
		\includegraphics[height=3cm, width=7cm]{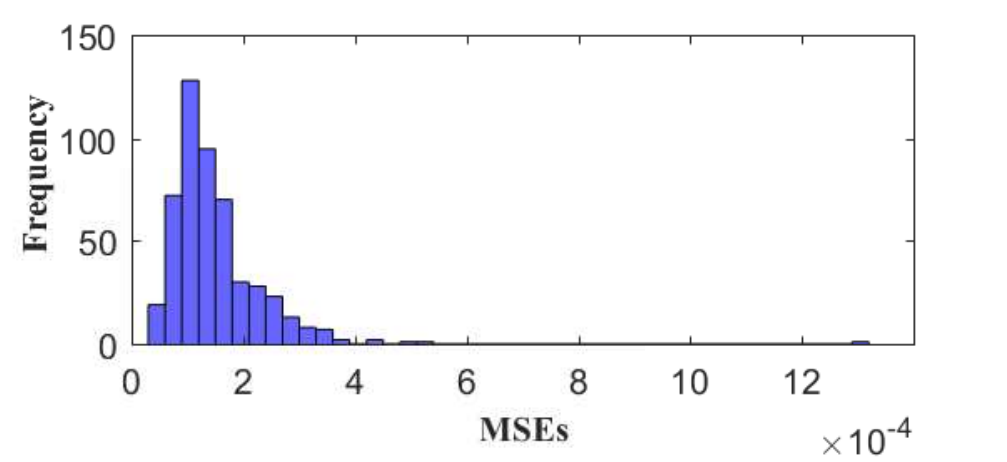}
		\includegraphics[height=3cm, width=7cm]{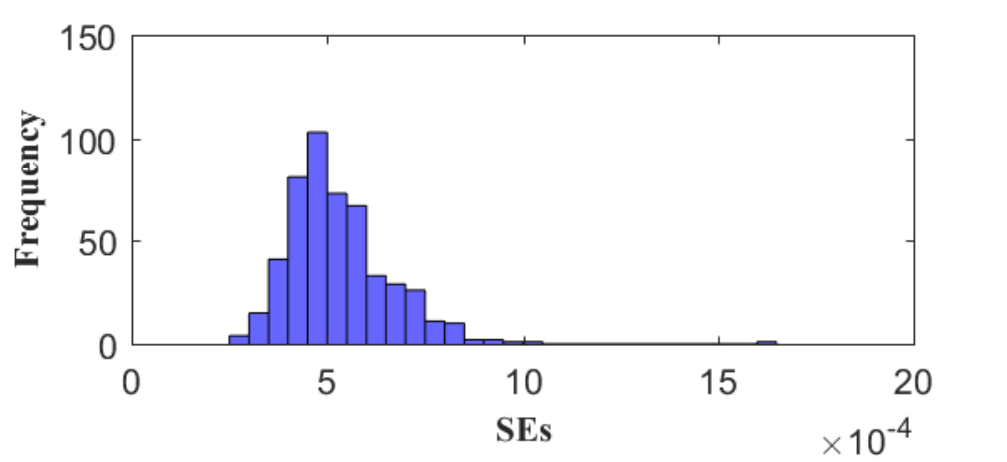}
		\includegraphics[height=3cm, width=7cm]{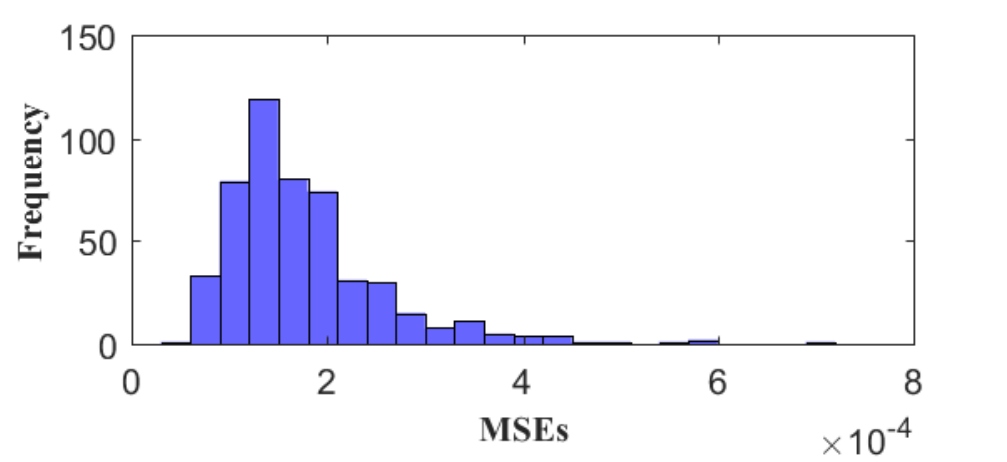}
		\includegraphics[height=3cm, width=7cm]{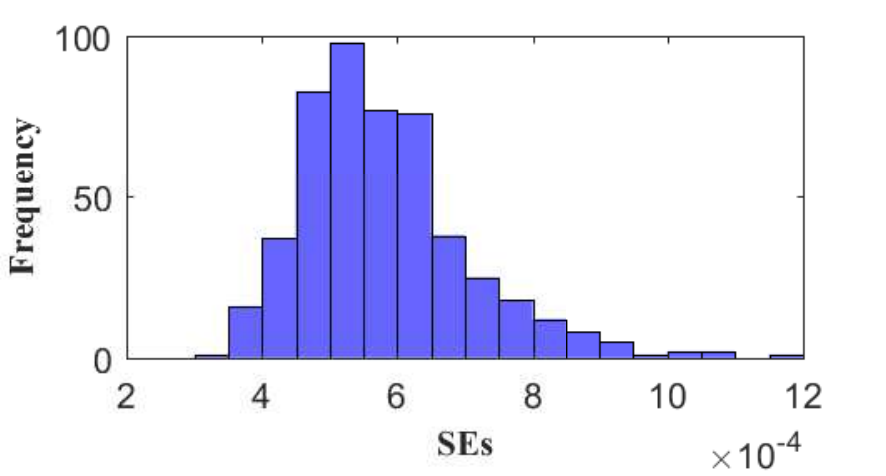}
		\includegraphics[height=3cm, width=7cm]{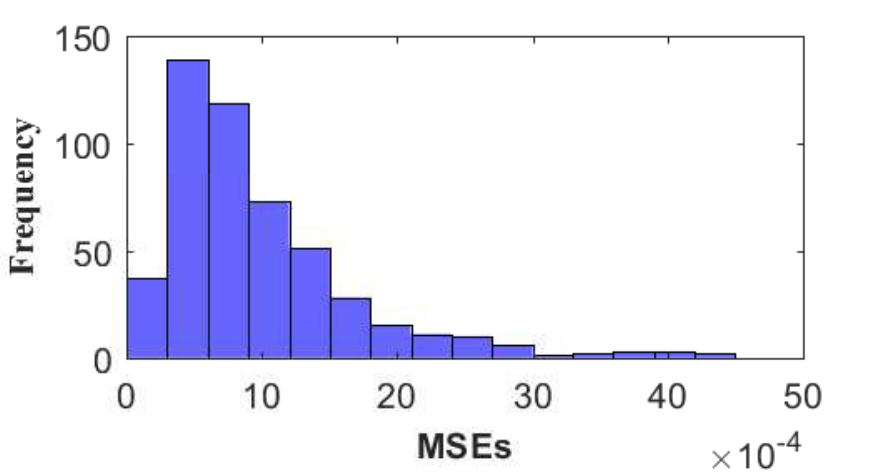}
		\includegraphics[height=3cm, width=7cm]{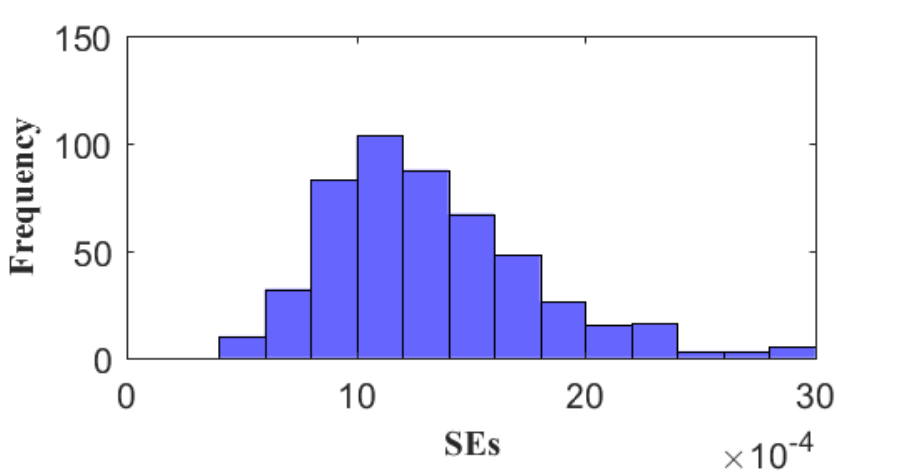}
		\normalsize\caption{Full histograms of MSEs and SEs for the AADMM  (top row), the AMP (middle row), and the ICR (bottom row) algorithms.}
		\label{Or223345}
	\end{figure}
The full histograms reveal that AADMM has better sparse signal recovery performance than AMP and ICR; note the differences in scales on the x-axis and the relative locations of the histogram bins. 
	
We  examined the convergence of the AADMM algorithm numerically for both unconstrained and constrained (or non-negative) sparse signal recovery.  Figure \ref{Or2267} shows the convergence of the AADMM algorithm.
	\begin{figure}[H]
		\centering
		\includegraphics[height=4cm, width=7cm]{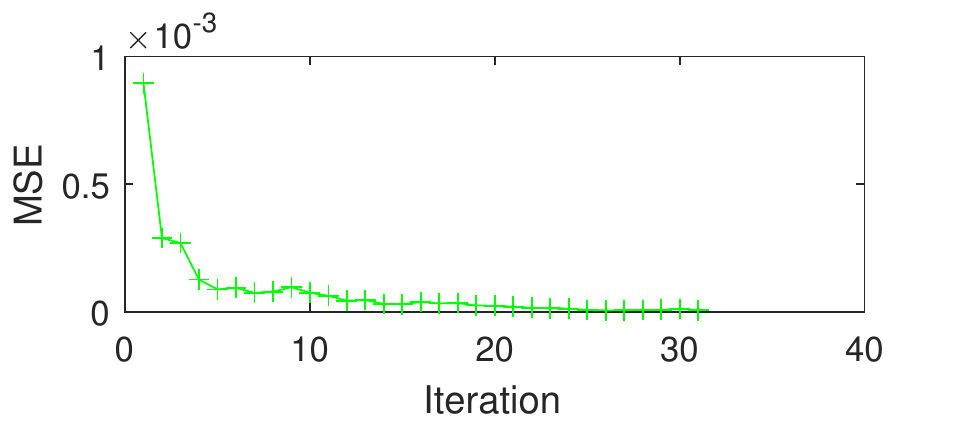}
		\includegraphics[height=4cm, width=7cm]{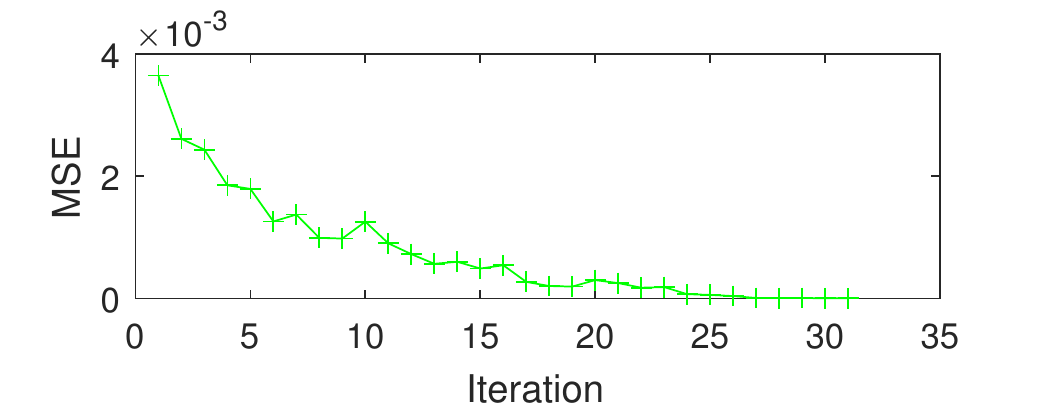}
		\normalsize\caption{Convergence of AADMM algorithm for unconstrained (left) and constrained (right) signal. MSE against the number of iterations is utilised to assess the convergence of the algorithm.}
		\label{Or2267}
	\end{figure}
	In addition, we have evaluated the effect of $\lambda$ on the sparse signal reconstruction performance of our algorithm, see Figure \ref{Or226rte7}.
	\begin{figure}[H]
		\centering
		\includegraphics[height=4cm, width=7cm]{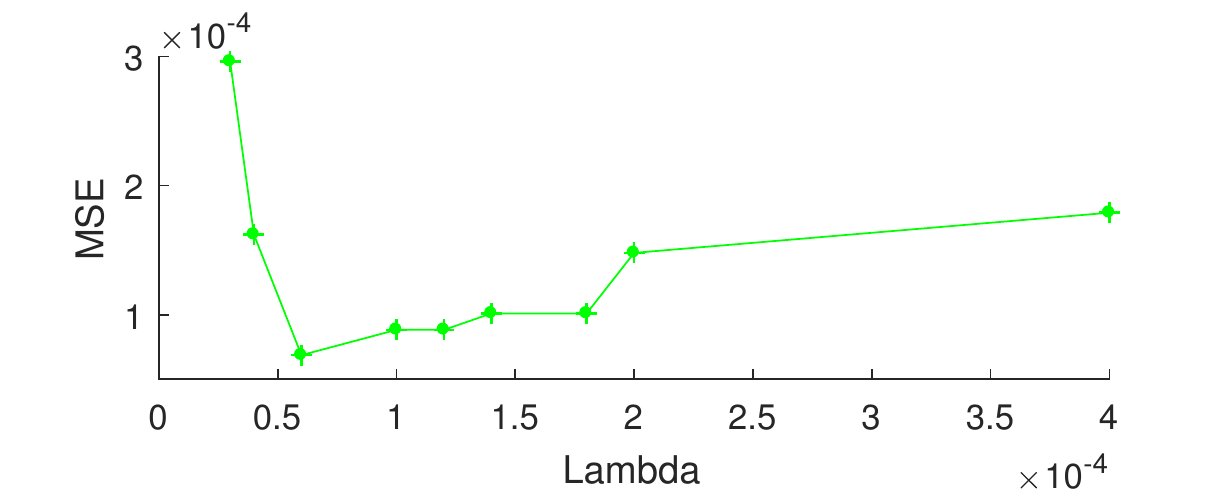}
		\includegraphics[height=4cm, width=7cm]{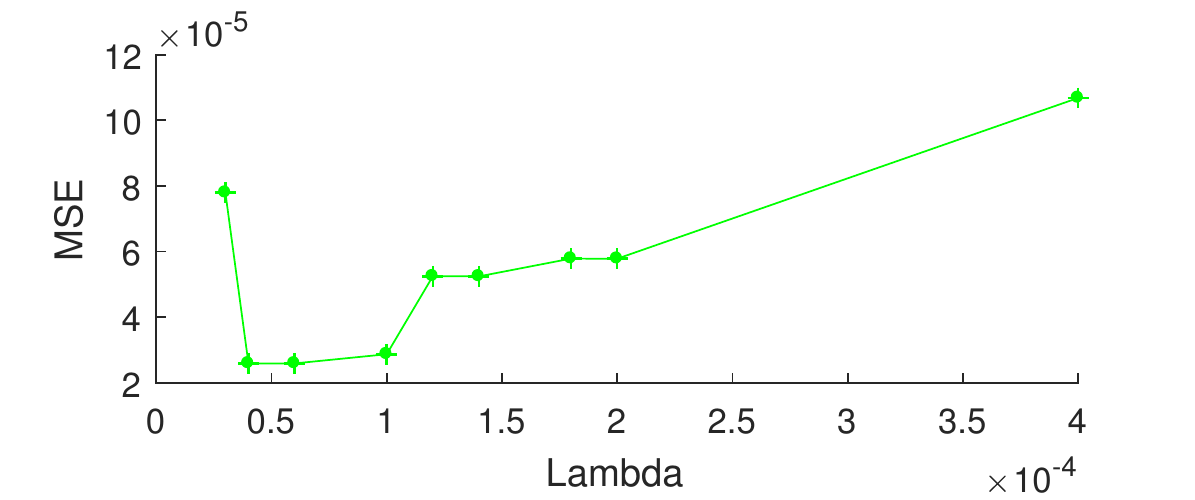}
		\normalsize\caption{Effect of $\lambda$ on the  sparse signal reconstruction performance of AADMM for unconstrained (left) and constrained (right) signal.}
		\label{Or226rte7}
	\end{figure}
	As expected, the  sparse signal reconstruction performance of AADMM depends on the regularizing parameter $\lambda$ and the reconstruction error can be minimized with a properly chosen $\lambda$. Note that we have compared the performance of ICR, AMP, and AADMM under parameter settings similar to those in \citep{YenT, BeckAandTeboulleM, MousaviHS,VuTH}.
	\subsection{Real image Recovery}
	We applied our algorithm to a signal with non-negative constraints. The non-negative constraint assumption on the signal allows us to explicitly enforce a non-negative constraint during the reconstruction of the signal. We utilised the algorithms on the well-known handwritten digit images MNIST \citep{LeCunY}. The digit images are real data, naturally sparse and fit into the spike and slab model. Each of the digit images (0 to 9) has a size of  28 $\times$ 28. The maximum and the minimum sparsity levels of the digit images are 200 and 96, respectively, which shows that the images are sparse.  We are interested to recover a sparse signal  $\mathbf{x} \in {\Bbb R}^{784\times 1}$ from undersampled random measurement $\mathbf{y}$.  For a Gaussian random matrix $\mathbf{A}$ and the maximum sparsity level 200, we can approximately determine the length of the random measurement $\mathbf{y}$ required for the successful recovery of a sparse signal $\mathbf{x}$ \citep{FoucartRauhut}.  Based on this, we randomly generated a  Gaussian random  matrix  $\mathbf{A}\in {\Bbb R}^{550\times 784}$ and an additive Gaussian noise  with  variance  $\sigma^2 = 3.24\times10^{-4}$ to obtain a random  measurement $\mathbf{y}$ according to equation \eqref{ABOq}.  Using $\lambda = 2\times10^{-4}$, the results obtained for the signal recovery problem (the images of the digits) and the evaluation of the algorithms can be found in Figure \ref{OrR2} and Table \ref{tab:AA223}, respectively. 
	\begin{figure}[H]
		\centering
		\includegraphics[height=10cm, width=15cm]{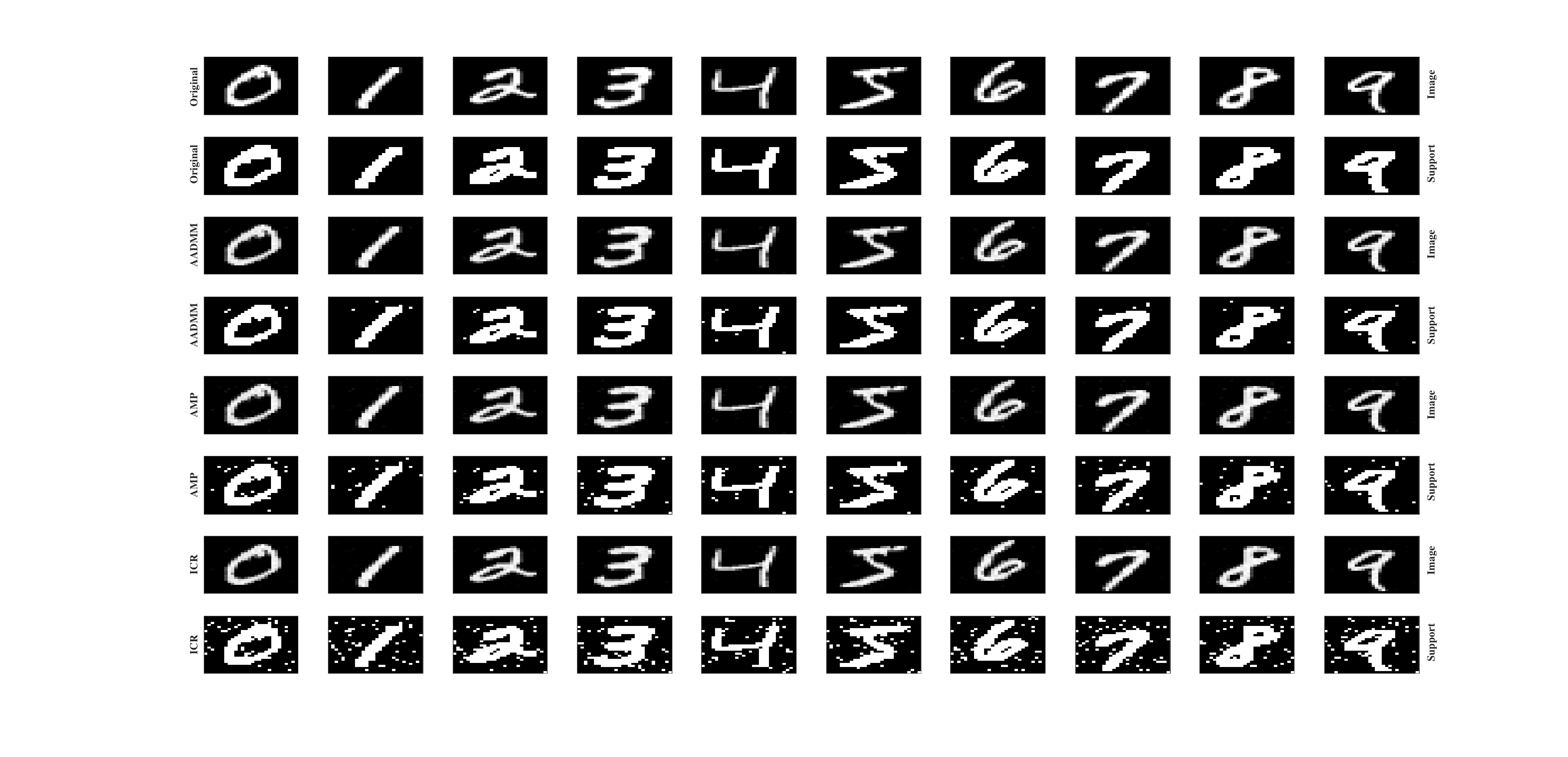}
		\vspace{-0.35in}
		\normalsize\caption{Sparse signal recovery of the real image using AADMM, AMP, and ICR. The first and the second rows show images of the original data and their supports. The third to the eighth rows present the recovered images and  their supports by AADMM, AMP, and  ICR, respectively.}
		\label{OrR2}
	\end{figure}
	\begin{table}[H]
		\centering
		\normalsize\caption{Evaluation of the algorithms for real image recovery. The averages of MSEs, SMLs, and CTs are employed to assess the performance of the algorithms for real image recovery.}
		\label{tab:AA223}
		\begin{tabular}{c|ccccc}
			\hline \hline
			Method&MSE&SML(\%)&CT (S)\\\hline
			AADMM&1.32$\times 10^{-4}$&98.495&1.020\\
			AMP&1.74$\times 10^{-4}$&97.360&1.228\\
			ICR &1.46$\times 10^{-4}$&94.133& 8.734\\\hline\hline
		\end{tabular}
	\end{table}
	
	Figure \ref{OrR2} shows the images of the original data and their supports (the first and the second rows), images of the recovered signal and their supports by AADMM (the third and the fourth rows), images of the  reconstructed signal and their supports by AMP (the fifth and the sixth rows), and images of the  reconstructed signal and their supports by ICR (the last two rows). The images of the supports give better information about the performance of the sparse signal recovery algorithms. Based on the images of the supports, it is clear that AADMM has better performance than AMP and ICR. Table \ref{tab:AA223} presents the average of MSEs, SMLs, and CTs of the reconstruction problem. It can be seen that AADMM outperforms AMP and ICR and it is also important to note that AADMM is faster than AMP and, in particular, much faster than ICR.
	%\section{Future Work} \label{FW}
	%Let $\mathbf{x}_{r}$ be a vectorized reconstructed MR image of interest, $\psi$ denote the linear operator that transforms the image into a sparse representation.  Compressibility or sparse representation using an appropriate sparsifying transform such as wavelet transform and finite-differencing is the key feature of MR images \citep{LustigMandDonohoDL}. Let $\mathcal{F}_{u}$ be the undersampled Fourier transform, corresponding to one of the $k$-space undersampling schemes and $\mathbf{y}$ be the measured k-space data from MRI scanner. 
	%MR image can reconstructed by solving the following
	%constrained optimization problem:
	%\[
	%\min \limits_{\mathbf{x}_{r}}\lVert\psi\mathbf{x}_{r}\lVert_{1} \quad \text{such that} \quad \lVert \mathcal{F}_{u}\mathbf{x}_{r}-\mathbf{y}\rVert^{2}_{2} \leq\boldsymbol\epsilon,
	%\]
	%\citep{LustigMandDonohoD}.
	%
	%To use our approach, which is according to the model setting in  equation \eqref{ABO1}, in MR image reconstruction, we set
	% \begin{eqnarray*}
	%\mathbf{x} = \psi\mathbf{x}_{r},\quad \mathbf{A}\mathbf{x} = \mathbf{A}\psi\mathbf{x}_{r} = \mathcal{F}_{u}\mathbf{x}_{r},
	%\end{eqnarray*}
	%where $\mathcal{F}_{u} = \mathbf{A}\psi$. Using this setting, we will investigate MR image reconstruction in the future. 
	\section{Conclusions}\label{conc}
	In this paper, we have developed  an algorithm (AADMM) to optimize a hard non-convex optimization problem and applied it in sparse signal recovery. Unlike the recent algorithm (ICR), AADMM does not simplify the optimization by considering a history of
	solutions at previous iterations. That is, AADMM solves the problem in its general form.  The most recent AMP algorithm  can not be used to directly solve this problem, which means that AADMM is a more general problem solving algorithm.  Our evaluation of the algorithm on simulated data and real-world image data shows that AADMM has superior practical merit on ICR.  Further, AADMM has better performance than AMP, in particular with respect to obtaining sparser signal recovery. 
	
	Regarding future work, one part is to adopt the proposed AADMM algorithm into the compressive sensing MR image reconstruction framework, where the MR images are not sparse in themselves but sparse under a specific transformation. The key issue in MR image reconstruction is to obtain sparser recovery of MR images, which is the main reason for designing our algorithm. A further reason is to generalize the spike and slab prior to incorporate structural sparsity for sparse signal recovery and to develop its corresponding sparse signal recovery algorithm. 
	\section*{Acknowledgments}
	\noindent This work is supported by the Swedish Research Council grant (Reg. No. 340-2013-5342). We would like to thank three anonymous referees and the Editor for their detailed and insightful comments and suggestions that help to improve the quality of this paper. 
	\section*{Disclosure of Conflicts of Interest}
	The authors have no relevant conflicts of interest to disclose.
	\section*{References}
	\bibliographystyle{apalike}
	\bibliography{sample_4}
\end{document}